\documentclass[11pt,a4paper]{article}
\usepackage{fullpage}

\usepackage{amsmath}
\usepackage{amsfonts}
\usepackage{amssymb}
\usepackage{amsthm}
\usepackage{amsbsy}
\usepackage{graphicx}
\usepackage{subcaption}
\usepackage{verbatim}
\usepackage{url}
\usepackage{bbm}
\usepackage{multirow}

\usepackage{hyperref}
\usepackage[svgnames]{xcolor}
\usepackage[capitalise,nameinlink]{cleveref}
\hypersetup{colorlinks={true},linkcolor={DarkBlue},citecolor=[named]{DarkGreen}}

\crefname{ineq}{Ineq.}{Inequalities}
\creflabelformat{ineq}{#2{\upshape(#1)}#3}

\crefname{cond}{Condition}{Conditions}
\creflabelformat{cond}{#2{\upshape(#1)}#3}

\usepackage{tcolorbox}

\usepackage{bm}
\usepackage{tikz}
\usepackage{ctable}

\usepackage[ruled]{algorithm2e}

\SetArgSty{textrm}  

\newtheorem{theorem}{Theorem}
\newtheorem{lemma}[theorem]{Lemma}

\theoremstyle{definition}
\newtheorem{example}{Example}

\usepackage{libertine}

\newcommand\vc{\mathbf}
\newcommand{\E}{\mathbb{E}}
\newcommand{\ud}{\,\mathrm{d}}

\usepackage[sort&compress,square,comma,authoryear]{natbib}

\newcommand{\SW}{\text{SW}}

\newcommand{\calT}{\mathcal{T}}

\begin{document}
\title{\bf Truthful ownership transfer with expert advice:\\Blending mechanism design with and without money}
\author{%
	\makebox[.3\linewidth]{Ioannis Caragiannis\thanks{University of Patras, Greece. E-mail: caragian@ceid.upatras.gr}}
	\and \makebox[.3\linewidth]{Aris Filos-Ratsikas\thanks{University of Liverpool, United Kingdom. E-mail: Aris.Filos-Ratsikas@liverpool.ac.uk}}
	\and \makebox[.3\linewidth]{Swaprava Nath\thanks{Indian Institute of Technology Kanpur, India. E-mail:  swaprava@cse.iitk.ac.in}}
	\and \makebox[.3\linewidth]{Alexandros A. Voudouris\thanks{University of Oxford, United Kingdom. E-mail: Alexandros.Voudouris@cs.ox.ac.uk}}
}
\date{}

\maketitle

\begin{abstract}
When a company undergoes a merger or transfers its ownership, the existing governing body has an opinion on which buyer should take over as the new owner. Similar situations occur while assigning the host of big sports tournaments, like the World Cup or the Olympics. In all these settings, the values of the external bidders are as important as the opinions of the internal experts. Motivated by such scenarios, we consider a social welfare maximizing approach to design and analyze truthful mechanisms in {\em hybrid social choice} settings, where payments can be imposed to the bidders, but not to the experts. Since this problem is a combination of mechanism design with and without monetary transfers, classical solutions like VCG cannot be applied, making this a novel mechanism design problem. We consider the simple but fundamental scenario with one expert and two bidders, and provide tight approximation guarantees of the optimal social welfare. We distinguish between mechanisms that use ordinal and cardinal information, as well as between mechanisms that base their decisions on one of the two sides (either the bidders or the expert) or both. Our analysis shows that the cardinal setting is quite rich and admits several non-trivial randomized truthful mechanisms, and also allows for closer-to-optimal welfare guarantees. 
\end{abstract}

\maketitle

\section{Introduction}\label{sec:intro}
Most well-studied problems in \emph{computational social choice}~\citep{BCEL+16} deal with combining individual preferences over alternatives 
into a collective choice. More often than not, the mechanisms employed for this aggregation task are {\em ordinal}, i.e., they do not use the intensities of the preferences of the individuals, and {\em non-truthful}, which is justified by several impossibility theorems \citep{Gib73,Satterthwaite75,Gib77}. On the other hand, the class of truthful \emph{cardinal} mechanisms has been shown to be much richer \citep{Freixas84,Barbera98,Feige10} and the additional information provided by the numerical values can notably increase the well-being of society \citep{GC:10,filos2014truthful,analytic}.
At the same time, truthful mechanisms \emph{with money} are pretty well-understood and  welfare-maximizing mechanisms for a wide class of problems are known \citep{AGT}. 

However, in a rich set of problems, where monetary transfers are possible only for some participants, designing {\em truthful, cardinal} mechanisms is more challenging -- one needs to combine elements of {\em mechanism design with money} and {\em social choice}. 
In this work, we consider such a setting where the agents are partitioned into two types, such that some of them offer monetary compensations, while other do not. 
The objective is to make a decision that maximizes the {\em social welfare}, which includes the {\em cardinal} values of both types of agents. This is a {\em hybrid social choice} setting, which blends together classical social choice and classical mechanism design with money, but is distinct from both of them, thereby rendering celebrated solutions like the VCG mechanism \citep{vickrey1961counterspeculation,clarke1971multipart,groves1973incentives} insufficient.

Let us provide a few examples of such hybrid social choice scenarios. Government agencies routinely sell public assets such as spectrum, land, or government securities, by transferring their ownership (or usage rights) to interested buyers. As such transfers may have huge impact to citizens, the decision about the new ownership is not simply the outcome of some competitive process among the potential buyers (e.g., through an auction), but usually also involves experts from the citizen community who provide advice regarding the societal impact of each potential ownership transfer \citep{janssen2004auctioning,Indi2018}. In contrast to each potential buyer who faces a value-for-money trade-off, the experts care only about societal value; their compensation is unrelated to the ownership decision and instead depends on their reputation and experience only. The government needs both parties for a successful transfer of the public assets and a reasonable goal would be to maximize the {\em social welfare}, which aggregates the values of buyers and experts for the ownership transfer. Furthermore, in the organization of sporting events, the bids of the potential hosts are taken into consideration along with the recommendations of a respective sports' administrative body (e.g., IOC for Olympic Games, FIFA for the World Cup, FIA for Formula One, etc.). 

\subsection{Our contribution and techniques}\label{sec:contributions}
We study a very simple but fundamental \emph{hybrid social choice} setting with two competing {\em bidders} $A$ and $B$, and a single {\em expert} with cardinal preferences over the three options of selling to bidder $A$, selling to bidder $B$, or not selling at all. 
A mechanism takes as input the bids and the expert's preferences, and decides one of the three options as outcome. In general, mechanisms are {\em randomized}; for a given input, they select the outcome using a probability distribution (or {\em lottery}) over the three options.

We consider mechanisms that can be {\em implemented truthfully}. 
Besides the outcome, the mechanism also outputs payments which are imposed to the bidders. 
The lottery and the payments should be such that
\begin{itemize}
\item the expert is incentivized to report her true preferences in order to maximize her (expected) value for the outcome, and
\item the bidders are incentivized to report their true values as bids in order to maximize their utility, i.e., their expected value for the outcome minus their payment to the mechanism.
\end{itemize}
In the following, we refer to mechanisms with such implementations as {\em truthful mechanisms}.

Interestingly, the theory of mechanism design allows us to abstract away from payments and view truthful mechanisms simply as lotteries. Well-known characterizations for single-parameter mechanism design with money from the literature, as well as new characterizations that we prove here for lotteries that guarantee truthfulness from the expert's side, are the main tools we use in order to constrain the design space of truthful mechanisms in our setting. 

Additional informational restrictions can further divide truthful mechanisms into the following classes:
\begin{itemize}
\item \emph{ordinal} mechanisms, which ignore the exact bids and the expert's preference values and instead take into account only their relative order,
\item \emph{bid-independent} mechanisms, which ignore the bids and base their decision solely on the expert's cardinal preferences,
\item \emph{expert-independent} mechanisms, which ignore the expert's preferences and base their decision solely on the bids, and
\item general truthful mechanisms, which may take {\em both} the bids and the expert's preference values into account.
\end{itemize}

We measure the quality of truthful mechanisms in terms of the {\em social welfare}, i.e., the aggregate value of the bidders and the expert for the outcome. Unfortunately, our setting does not allow for a truthful implementation of the social welfare-maximizing outcome. So, we resort to near-optimal truthful mechanisms and use the notion of the {\em approximation ratio} to measure their quality. Even though the setting that we study seems simple, it turns out that identifying the best possible truthful mechanism for the several classes mentioned above is a challenging task, and the mechanisms themselves, as well as their analyses, are often quite involved.

For the classes of ordinal, bid-independent, and expert-independent mechanisms, we prove lower bounds on the approximation ratio of truthful mechanisms in the class, and identify the best possible among them, with approximation ratios of $1.5$, $1.377$, and $1.343$, respectively. Furthermore, by slightly enhancing expert-independent mechanisms and allowing them to utilize a {\em single bit} of information about the expert's preferences, we define a template for the design of new truthful mechanisms. The template defines {\em always-sell} mechanisms that select either bidder A or bidder B as the outcome. We present two mechanisms that follow our template, one deterministic and one randomized, with approximation ratios $1.618$ and $1.25$, respectively. The former is best-possible among all deterministic truthful mechanisms. The latter is best-possible among all always-sell truthful mechanisms. We also present an unconditional lower bound of $1.141$ on the approximation ratio of any truthful mechanism. These results are summarized in \cref{tab:results}.

\noindent
\begin{table}[h]
\centering
\begin{tabular}{lcl}
\noalign{\hrule height 1pt}\hline
	class of mechanisms & approximation ratio 	& reference\\ 
\noalign{\hrule height 1pt}\hline
	ordinal 			& $1.5$ 				& mechanisms EOM, BOM (\cref{thm:EOM})	 \\
					   	& 						& best possible (\cref{thm:ordinal-lower})	 \\
	bid-independent 	& $1.377$ 				& mechanism BIM (\cref{thm:bim})\\
	 					&  						& best possible (\cref{thm:bim-lower})\\
	expert-independent 	& $1.343$ 				& mechanism EIM (\cref{thm:eim})\\
	 					&  						& best possible (\cref{thm:eim})\\
	our template 		& $1.25$ 				& randomized mechanism $R$ (\cref{thm:R})\\
						&  						& best possible, always-sell (\cref{thm:always-sell})\\
						& $1.618$ 				& deterministic mechanism $D$ (\cref{thm:R})\\
						& 						& best possible, deterministic (\cref{thm:D-lower})\\
	all mechanisms 		& $1.141$ 				& unconditional lower bound (\cref{thm:lower})\\
\noalign{\hrule height 1pt}\hline
\end{tabular}
\caption{Overview of our results. Unless specified otherwise, the term ``best possible'' means best possible among the mechanisms in the particular class of mechanisms.}\label{tab:results}
\end{table}
	
Both our positive and negative results have been possible by narrowing the design space using truthfulness characterizations, the particular structure in each class of mechanisms, as well as the goal of low approximation ratio. In most cases, by carefully blending together all these factors, the design of new mechanisms turns out to be as simple as drawing a curve in a restricted area of a $2$-dimensional plot (e.g., see \cref{fig:eim} and \cref{fig:template}).

\subsection{Related work} \label{sec:related}
Our setting can be viewed as an instance of \emph{approximate mechanism design}, \emph{with} \citep{nisan2001algorithmic} and \emph{without money} \citep{PT09}, which was proposed for problems where the goal is to optimize an objective under the strict truthfulness requirement.
A result that will be very useful to our analysis is Myerson's characterization for single-parameter domains~\citep{myerson1981optimal}, which provides necessary and sufficient conditions for (deterministic or randomized) mechanisms (with money) to be truthful. It allows us to abstract away from the payment functions (which are uniquely determined given the selection probabilities) on the bidders' side. Furthermore, similar arguments based on the same characterization enable us to reason about the structure of truthful mechanisms (without money) on the expert's side as well.

When monetary transfers are allowed, the well-known Vickrey-Clarke-Groves (VCG) mechanism \citep{vickrey1961counterspeculation,clarke1971multipart,groves1973incentives} is deterministic, truthful, and maximizes the social welfare in many settings of interest. However, as we pointed out in the discussion above, in our hybrid mechanism design setting, one needs to take the values of the expert into account as well, and therefore VCG is no longer truthful nor optimal. On the expert's side, truthful mechanisms can be thought of as truthful voting rules; any positive results for deterministic such rules are impaired by the celebrated Gibbard-Satterthwaitte impossibility theorem \citep{Gib73,Satterthwaite75} which limits this class to only \emph{dictatorial mechanisms}. 

In contrast, the class of randomized truthful voting rules is much richer and includes many reasonable truthful rules that are not dictatorial. In fact, \citet{Gib77} characterized the class of all such ordinal rules; a general characterization of all cardinal voting rules is still elusive. To this end, a notable amount of work in the classical economics literature as well as in computer science has been devoted towards designing such rules and proving structural properties for restricted classes. \citet{Gibbard78} provided a similar characterization to his 1977 result, which however only holds for discrete strategy spaces, and later \citet{hylland1980strategy}\footnote{Quite remarkably, this paper is unpublished -- the result was revisited by \citet{dutta2007strategy}.} proved that the class of truthful rules that are Pareto-efficient reduces to random dictatorships. \citet{Freixas84} used the differential approach to mechanism design, proposed by \citet{laffont1980differential}, to design a class of truthful mechanisms which actually characterizes the class of twice differentiable mechanisms over subintervals of the valuation space; the best-possible truthful bid-independent mechanism that we propose in this paper can be seen as a mechanism in this class. \citet{Barbera98} showed that there are many interesting truthful mechanisms which do not fall into the classes considered by \citet{Freixas84}. In the computer science literature, \citet{Feige10} designed a class of one-voter cardinal truthful mechanisms, where the election probabilities are given by certain polynomials. 

Social welfare maximization without payments has been studied in a plethora of related papers in the computer science literature, in general social choice settings \citep{filos2014truthful,bhaskar2018}, as well as in restricted domains, such as matching and allocation problems \citep{filos2014social,analytic,GC:10}. Similarly to what we do here, \citet{filos2014truthful} make use of one-voter cardinal truthful mechanisms to achieve improved welfare guarantees. However, the presence of the bidders significantly differentiates our setting from theirs (as well as the other related works), since we have to consider both sides in the design and analysis of mechanisms. Another relevant notion is that of the \emph{distortion} of (non-truthful) mechanisms which operate under limited information (typically ordinal mechanisms) \citep{CP11,BCHL+15,CNPS16,CPS16,ABEPS18,ABFV19}. While the lack of information has also been a restrictive factor for some of our results (in conjunction with truthfulness), we are mainly interested in cardinal mechanisms for which truthfulness is the limiting constraint.

\subsection{Roadmap}
The rest of the paper is structured as follows. We begin with preliminary definitions, notation and examples in \cref{sec:prelim}. Then, \cref{sec:ordinal}, \cref{sec:bim}, and \cref{sec:eim} are devoted to ordinal, bid-independent and expert-independent mechanisms, respectively. Our template and the corresponding best possible deterministic and randomized mechanisms are presented in \cref{sec:beyond}, while our unconditional lower bounds are presented in \cref{sec:lower-bounds}. We conclude with a discussion of possible extensions and open problems in \cref{sec:open}. Due to lack of space, some proofs appear in appendix.


\section{Preliminaries}\label{sec:prelim}
Our setting consists of two agents $A$ and $B$ who compete for an item and an expert $E$. The agents have valuations $w_A,w_B \in \mathbb{R}_{\geq 0}$ denoting the (maximum) amount of money that they would be willing to spend for the item, and the expert has a valuation function $v: \mathcal{O} \rightarrow \mathbb{R}_{\geq 0}$ over the following three {\em options}: agent $A$ is selected to get the item, or agent $B$ is selected, or no agent is selected to get the item. We use $\oslash$ to denote this last option; hence, $\mathcal{O}=\{A,B,\oslash\}$. We use $\vc{w}=(w_A,w_B)$ to denote an \emph{agent profile} and let $\mathcal{W}$ be the set of all such profiles. Similarly, we use $\vc{v}=(v(A), v(B), v(\oslash))$ to denote an \emph{expert profile} and let $\mathcal{V}$ be the set of all such profiles. The domain of our setting is $\mathcal{D} = \mathcal{V} \times \mathcal{W}$. From now on, we use the term {\em profile} to refer to elements of $\mathcal{D}$.

A mechanism $M$ takes as input a profile $(\vc{v},\vc{w})$ and decides, according to a probability distribution (or {\em lottery}) $P^M$ a pair $(o,\vc{p})$ consisting of an option $o\in \mathcal{O}$ and a vector $\vc{p}=(p_A,p_B)$ indicating the payments that are imposed to the agents. The execution of the mechanism 
yields a utility to the expert and the agents. Given an outcome $(o,\vc{p})$ of the mechanism, the utility of the expert is $u_E(o,\vc{p}) = v(o)$; the utility of agent $i \in \{A,B\}$ is $u_i(o,\vc{p})=w_i-p_i$ if $i=o$ and $u_i(o,\vc{p})=-p_i$ otherwise.

The mechanism asks the expert to submit a {\em report} and the agents to submit their {\em bids}. All of them however may have incentive to misreport their true values in order to maximize their utility.
We are interested in mechanisms that do not allow for such strategic manipulations. We say that a mechanism $M$ is \emph{truthful for agent} $i \in \{A,B\}$ if for any value $w_i$ and any profile $(\vc{v}',\vc{w}')$, 
$$\E[u_i(M(\vc{v}',(w_i,{w'}_{-i}))]\geqslant \E[u_i(M(\vc{v}',\vc{w}'))],$$ 
where the expectation is taken with respect to the lottery $P^M$. This means that bidding her true value $w_i$ is a utility-maximizing strategy for the agent, no matter what the other agent bids and the expert reports. Similarly, mechanism $M$ is said to be \emph{truthful for the expert} if for any expert profile $\vc{v}$ and any profile $(\vc{v}',\vc{w}')$, 
$$\E[u_E(M(\vc{v},\vc{w}'))]\geqslant \E[u_E(M(\vc{v}',\vc{w}'))].$$ 
Again, this means that reporting her true valuation profile is a utility-maximizing strategy for the expert, no matter what the agents bid. A mechanism $M$ is \emph{truthful} if it is truthful for the agents and truthful for the expert.

Our goal is to design truthful mechanisms that achieve high social welfare, which is the total value of the agents and the expert for the outcome. 
 For a meaningful definition of the social welfare that weighs equally the {valuations of the expert and the agents, we adopt a canonical representation of profiles. The expert has 
{\em normalized}
{\em von Neumann-Morgenstern valuations}, i.e., her values for two of the options are $0$ and $1$, while her value for the third option lies in the interval $[0,1]$. The values of the agents are normalized in the definition of the social welfare, which is defined as
\begin{align*}
\SW(o,\vc{v},\vc{w}) = 
\begin{cases}
v(o) +\frac{w_o}{\max\{w_{A},w_{B}\}}, & \text{if } o \in \{A,B\} \\
v(\oslash), & \text{otherwise}.
\end{cases}
\end{align*}
We measure the quality of a truthful mechanism $M$ by its {\em approximation ratio}, which (by abusing notation a bit and interpreting $M(\vc{v},\vc{w})$ as the option decided by the mechanism) is defined as
\begin{align*}
\rho(M)& = \sup_{(\vc{v},\vc{w})\in \mathcal{D}}{\frac{\max_{o\in \mathcal{O}}{\SW(o,\vc{v},\vc{w})}}{\E[\SW(M(\vc{v},\vc{w}),\vc{v},\vc{w})]}}.
\end{align*}
Of course, low values of $\rho(M)$, as close as possible to $1$, are most desirable.

\subsection{An alternative view of profiles and mechanisms}
In order to simplify the exposition of our results in the following sections, we devote some space here to introduce two alternative ways of representing profiles, which we call the {\em expert's view} and the {\em agents' view}.
Without essentially restricting the space of mechanisms that can achieve good approximation ratios according to our definition of social welfare, we focus on mechanisms that base their decisions on the {\em normalized bid} values, i.e., on the quantities $\frac{w_A}{\max\{w_{A},w_{B}\}}$ and $\frac{w_B}{\max\{w_{A},w_{B}\}}$. It will be convenient to use the following two alternative ways
$$\left(\begin{array}{c c c}1 & x& 0\\h_a & \ell_a & z_a\end{array}\right) \mbox{ and } \left[\begin{array}{c c c}h_E & \ell_E& n_E\\1 & y & 0\end{array}\right]$$
to represent a profile $(\vc{v},\vc{w})$. The first representation is the expert's view, and the second one is the agents' view. Each column corresponds to an option. 
\begin{itemize}
\item According to the expert's view on the left, the columns are ordered in terms of the values of the expert, which appear in the first row. The quantities $h_a$, $\ell_a$, and $z_a$ hold the normalized agent bids for the corresponding option and $0$ for option $\oslash$. Essentially, $h_a$ is the value of the expert's favorite option, which can be equal to $1$ if it corresponds to the value of the agent with the highest value ({\em high-bidder}), equal to some value $y \in [0,1]$ if it corresponds to the value of the agent with the lowest value ({\em low-bidder}), or $0$ if it corresponds to the no-sale option $\oslash$. Similarly, $\ell_a$ and $z_a$ are the values of the expert's second and third favorite options, respectively. 

\item According to the agents' view on the right, the columns are ordered in terms of the bids, which appear in the second row. The quantities $h_E$, $\ell_E$, and $n_E$ now hold the valuations of the expert for the corresponding options. Essentially, $h_E$ is the value of the expert for the high-bidder, $\ell_E$ is the value of the expert for the low-bidder, and $n_E$ is the value of the expert for the no-sale option. All of them can take values in the interval $[0,1]$ such that one of them is equal to $1$ and another is equal to $0$. 
\end{itemize}
These representations yield a crisper way to argue about truthfulness for the expert and the agents in our main results. Specifically, in \cref{sec:bim}, we will study bid-independent mechanisms, and therefore it makes sense to use the expert's view of profiles, whereas in \cref{sec:eim}, it will be easier to argue about our expert-independent mechanisms based on the agents' view instead. The agents' view will also be used in \cref{sec:beyond}, where the mechanisms we present use the expert's opinion only to appropriately partition the input profiles into categories, and it is therefore easier to argue about their properties using the agents' view. In \cref{sec:lower-bounds} we will again use the expert's view to prove our unconditional lower bounds.

Similarly to the expert's and the agents' view described above, we use two different representations of the lottery $P^M$, depending on whether we represent profiles according to the expert's or the agents' view. From the expert's viewpoint, $P^M$ is represented by three functions $g^M$, $f^M$, and $\eta^M$, which correspond to the probability of selecting the first, second, and third favorite option of the expert, respectively. Similarly, from the agents' viewpoint, $P^M$ is represented by three functions $d^M$, $c^M$, and $e^M$, which correspond to the probability of selecting the high-bidder, the low-bidder, or option $\oslash$. 

In the upcoming sections, to simplify our discussion, we will sometimes drop $a$, $E$ and $M$ from our notation when the profile view and the mechanisms are clear from context.

\begin{example}\label{ex:views}
Let us present an example. Consider a profile with expert valuations $1$ for option $\oslash$, $0.3$ for option $A$ and $0$ for option $B$, and normalized bids of $1$ from agent $A$ and $0.9$ from agent $B$. Consider a lottery which for the particular profile uses probabilities $0.4$ for option $A$, $0.1$ for option $B$, and $0.5$ for option $\oslash$. 
The expert's and agents' views of the profile are
$$\left(\begin{array}{c c c}1 & 0.3& 0\\0 & 1 & 0.9\end{array}\right) \mbox{ and } \left[\begin{array}{c c c}0.3 & 0& 1\\1 & 0.9 & 0\end{array}\right],$$
respectively. 
The functions $g^M$, $f^M$ and $\eta^M$ are defined over the $4$-tuple of arguments $(x,h_a,\ell_a,z_a)=(0.3,0,1,0.9)$, which compactly represents the expert's view, and take values $0.5$, $0.4$, and $0.1$, respectively. Similarly, the functions $d^M$, $c^M$, and $e^M$ are defined over the $4$-tuple of arguments $(y,h_E,\ell_E,n_E)=(0.9,0.3,0,1)$, which compactly represents the agents' view, and take values $0.4$, $0.1$, and $0.5$, respectively. \hfill $\qed$
\end{example}

To handle ties in the expert's report or the agents' bids, we use the fixed priority $A\succ B\succ \oslash$ in order to identify the high- and low-bidder as well as the highest and lowest expert valuation. For example, if the expert has value $1$ for options $\oslash$ and $B$, we interpret this as option $B$ being her most favorite one. Similarly, when the bids are equal, agent $A$ is always the high-bidder and agent $B$ is the low-bidder. This is used in the definition of our mechanisms only; lower bound arguments do not depend on such assumptions in order to be as general as possible.

\subsection{Reasoning about truthfulness}
Let us now explain the truthfulness requirements having these profile representations in mind. There are two different kinds of possible misreports by the expert. She can attempt to make 
\begin{itemize}
\item a {\em level change in the reported valuation} (or ECh, for short) by changing her second highest valuation without affecting the order of her valuations for the options, or 
\item a {\em swap in the reported valuation} (ESw) by changing the order of her valuations for the options as well as the particular values. 
\end{itemize}
For example, the profile 
$$\left(\begin{array}{c c c}1 &0.6 &0\\0.9&0&1\end{array}\right)$$ 
is the result of a swap in the reported valuation by the expert who changes her valuations from $(1, 0.3, 0)$ (that she has in \cref{ex:views}) to $(0.6, 0, 1)$ for the three options $(\oslash, A, B)$. 

Similarly, there are also two different kinds of possible misreports by each agent. In particular, the agent can attempt to make 
\begin{itemize}
\item a {\em level change in the bid} (BCh) by changing her bid without affecting the order of bids or 
\item a {\em swap in the reported bid } (BSw) by changing both the bid order and the corresponding values. 
\end{itemize}
For example, the profile 
$$\left[\begin{array}{c c c}0 & 0.3& 1\\1 & 0.25 & 0\end{array}\right]$$ 
is the result of a swap in the reported bid by the low-bidder, who increases her bid in the profile of \cref{ex:views} to a new bid that is four times the bid of the other agent.

A truthful mechanism never incentivizes (i.e., it is {\em incentive compatible} with respect to) such misreports. We use the terms ECh-IC, ESw-IC, BCh-IC, and BSw-IC to refer to incentive compatibility with respect to the misreporting attempts mentioned above. A {\em truthful} mechanism, therefore, satisfies all these IC conditions. 
Before we proceed, we provide a few examples of truthful mechanisms. 

\begin{example}[A bid-independent ordinal mechanism]\label{ex:ordinalbiexample}
Consider the following mechanism that ignores the bids reported by the agents. With probability $2/3$ output the expert's favorite option, and with probability $1/3$ output the expert's second favorite option. Adopting the expert's view and the corresponding representation of the lottery $P^M$, the mechanism can be written as:
\begin{align*}
g^M(x,h_a,\ell_a,z_a) = \frac{2}{3}, \ \ \ \ f^M(x,h_a,\ell_a,z_a) = \frac{1}{3} \ \ \ \text{and} \ \ \ \eta^M(x,h_a,\ell_a,z_a)=0.
\end{align*} 
The mechanism can be seen to be truthful by the fact that (a) it ignores the bids of the agents and (b) it always assigns higher probability to the most-preferred outcome for the expert and $0$ probability to the least-preferred outcome. Note that using the terminology above, any ordinal mechanism is ECh by construction, since changing the level in the reported valuation does not change the outcome. \hfill $\qed$
\end{example}

\begin{example}[A bid-independent mechanism which is not ordinal] \label{ex:biexample}
Consider the expert's view (according to which  $x$ is the value of the expert for her second favorite outcome) and the corresponding representation of the lottery $P^M$, which is given by:
\begin{align*}
g^M(x,h_a,\ell_a,z_a) = \frac{4-x^2}{6}, \ \ \ \ f^M(x,h_a,\ell_a,z_a) = \frac{1+2x}{6} \ \ \ \text{and} \ \ \ \eta^M(x,h_a,\ell_a,z_a)=\frac{1-2x+x^2}{6}.
\end{align*} 
Note that this mechanism ignores the bids of the agents and uses the cardinal information reported by the expert. This mechanism has been referred to in the literature as the \emph{quadratic lottery} and has been proved to be truthful \citep{Feige10,Freixas84}.\hfill $\qed$
\end{example}	

\begin{example}[A deterministic expert-independent mechanism] \label{ex:eiexample}
Consider the following mechanism that ignores the expert's values for the different outcomes. 
Output the high-bidder and charge this agent a payment equal to the bid of the other agent.
Charge the other agent a payment of $0$. 
In terms of the agents' view, the outcome of the mechanism can be written as:
\begin{align*}
d^M(y,h_E,\ell_E,n_E) = 1, \ \ \ \ c^M(y,h_E,\ell_E,n_E) = 0 \ \ \ \text{and} \ \ \ e^M(y,h_E,\ell_E,n_E)=0.
\end{align*} 
This mechanism is the well-known \emph{second-price auction}~\citep{vickrey1961counterspeculation}, which is known (and easily seen) to be truthful. \hfill $\qed$
\end{example} 
It is not hard to see that none of the mechanisms presented in the above examples 
can achieve a very strong approximation ratio. As we will see in \cref{sec:ordinal}, the mechanism of \cref{ex:ordinalbiexample} is actually the best possible among the restricted class of ordinal mechanisms. Later on, the use of cardinal information will allow us to decisively outperform it. We also note that while the second-price auction in \cref{ex:eiexample} is welfare-optimal for the agents, which is a well-known fact, it  provides only a $2$-approximation when it comes to our objective of the combined welfare of the agents and the expert.

We continue with important conditions that are necessary and sufficient for BCh-IC and ECh-IC. The next lemma is essentially the well-known characterization of \citet{myerson1981optimal} for single-parameter domains.

\begin{lemma}[\citet{myerson1981optimal}]\label{lem:BCh-IC}
A mechanism $M$ is BCh-IC if and only if the functions $d^M$ and $c^M$ are non-increasing and non-decreasing in terms of their first argument (the low bid), respectively.
\end{lemma}

As long as the output of a mechanism satisfies the monotonicity condition of \cref{lem:BCh-IC}, one can always find payments for the agents that will make the mechanism BCh-IC. In fact, when the mechanisms are required to charge a payment of zero to an agent with a zero bid, then these payments are uniquely defined, and are given by the following formula
\begin{align*}
p_i(w_i,w_{-i}) = w_i\cdot q_i(w_i,w_{-i}) - \int_{0}^{w_i}q_i(t,w_{-i}) \ud t,
\end{align*} 
where $q_i$ is the probability that agent $i \in \{A,B\}$ will be selected as the outcome, $p_i$ is the payment function, $w_i$ is the bid of agent $i$ and $w_{-i}$ is the bid of the other agent. Therefore, we can avoid referring to the payment function when designing our mechanisms, as we can choose the above payment function, provided that the outcome probabilities satisfy the monotonicity conditions of \cref{lem:BCh-IC}. On the other hand, our lower bounds apply to all mechanisms, regardless of the payment function, as they only use the monotonicity condition. \\

\noindent Next, we provide a similar proof to that of \citet{myerson1981optimal} for characterizing ECh-IC in our setting.
\begin{lemma}\label{lem:ECh-IC}
A mechanism $M$ is ECh-IC if and only if the function $f^M$ is non-decreasing in terms of its first argument and the function $g^M$ satisfies 
\begin{align}\label{eq:lem:ECh-IC}
g^M(x,h_a,\ell_a,z_a) = g^M(0,h_a,\ell_a,z_a)-xf^M(x,h_a,\ell_a,z_a)+\int_0^x{f^M(t,h_a,\ell_a,z_a)\ud t},
\end{align}
for every $4$-tuple $(x,h_a,\ell_a,z_a)$ representing a profile as seen by the expert. 
\end{lemma}
\noindent As a corollary, functions $g^M$ and $h^M$ are non-increasing in terms of the first argument.

\begin{proof}
To shorten notation, we use $\vc{b}=(h_a,\ell_a,z_a)$ as an abbreviation of the information in the second row of a profile in expert's view and $(x,\vc{b})$ as an abbreviation of $(x,h_a,\ell_a,z_a)$. Also, we drop $M$ from notation (hence, $f(x,\vc{b})$ is used instead of $f^M(x,h_a,\ell_a,z_a)$) since it is clear from context. Due to ECh-IC, the expert has no incentive to attempt a level change of her valuation for her second favorite option from $x$ to $x'$. This means that
\begin{align}\label[ineq]{eq:x-to-xprime-deviation}
g(x,\vc{b})+xf(x,\vc{b}) &\geqslant g(x',\vc{b})+xf(x',\vc{b}).
\end{align}
Similarly, she has no incentive to attempt a level change of her valuation for her second favorite option from $x'$ to $x$. This means that
\begin{align}\label[ineq]{eq:xprime-to-x-deviation}
g(x',\vc{b})+x'f(x',\vc{b}) &\geqslant g(x,\vc{b})+x'f(x,\vc{b}).
\end{align}
By summing \cref{eq:x-to-xprime-deviation} and \cref{eq:xprime-to-x-deviation}, we obtain that
\begin{align*}
(x-x')(f(x,\vc{b})-f(x',\vc{b}))&\geqslant 0,
\end{align*}
which implies that $f$ is non-decreasing in terms of its first argument.
		
To prove \cref{eq:lem:ECh-IC}, we observe that \cref{eq:x-to-xprime-deviation} yields
\begin{align}
g(x,\vc{b})+xf(x,\vc{b}) &\geqslant g(x',\vc{b})+x'f(x',\vc{b})+(x-x')f(x',\vc{b}).
\end{align}
This means that function $g(x,\vc{b})+xf(x,\vc{b})$ is convex with respect to its first argument and has $f$ as its subgradient~\citep{rockafellar2015convex}. Hence, from the standard results of convex analysis we get
\begin{align*}
g(x,\vc{b})+xf(x,\vc{b}) &= g(0,\vc{b})+\int_0^x{f(t,\vc{b})\ud t},
\end{align*}
which is equivalent to \cref{eq:lem:ECh-IC}.
	
To verify that the conditions of the lemma are also sufficient for a mechanism to be ECh-IC, assume that $f$ is non-decreasing, $g$ satisfies \cref{eq:lem:ECh-IC}, and the expert has an incentive to make a level change of her valuation for her second favorite option from $x$ to $x'$. This means that
\begin{align*}
g(x,\vc{b})+xf(x,\vc{b}) < g(x',\vc{b})+xf(x',\vc{b}),
\end{align*}
which, by replacing $g$, is equivalent to 
\begin{align*}
(x'-x)f(x',\vc{b}) < \int_x^{x'} f(t,\vc{b}) \ud t,
\end{align*}
which contradicts the assumption that $f$ is non-decreasing. Hence, the expert does not have any incentive to make such level changes.
\end{proof}

We remark here that while \cref{lem:ECh-IC} will be fundamental for our proofs, it does not provide a characterization of all truthful one-voter mechanisms in the unrestricted social choice setting (such mechanisms are referred to as \emph{unilateral} in the literature). The reason is that (a) it applies only to changes in the intensity of the preferences and not swaps in the ordering of alternatives and (b) it only provides conditions for three alternatives, as opposed to many alternatives in the general setting. 


\section{Warmup: Ordinal mechanisms}\label{sec:ordinal}
We will consider several classes of truthful mechanisms depending on the level of information that they use. Let us warm up with some easy results on ordinal mechanisms, which do not use the exact values of the expert's report and the bids, but only their relative order. It turns out that the best possible approximation ratio of such mechanisms is $3/2$ and is achieved by two symmetric mechanisms, one depending only on the ordinal information provided by the expert (expert-ordinal), while the other depends only on the relation between the bids (bid-ordinal). 

The expert-ordinal mechanism EOM selects the expert's favorite option with probability $2/3$ and her second best option with probability $1/3$. 
Symmetrically, the bid-ordinal mechanism BOM selects the high-bidder with probability $2/3$ and the low-bidder with probability $1/3$. 
As we show next, both EOM and BOM are optimal among all ordinal mechanisms. 

\begin{theorem}\label{thm:EOM}
Mechanisms EOM and BOM are truthful and have approximation ratio at most $3/2$.
\end{theorem}

\begin{proof}
EOM is clearly truthful for the agents since it ignores the bids. It is also clearly truthful for the expert since the probabilities of selecting the options follow the order of the expert's valuations for them. BOM is clearly truthful for the expert (since her input is ignored); truthfulness for the agents follows by observing that the probability of selecting an agent is non-decreasing in terms of her bid.

We prove the approximation ratio for mechanism BOM only; the proof for EOM is completely symmetric. Consider the profile $\left[\begin{array}{c c c}h& \ell & n\\1 & y& 0\end{array}\right]$ in agents' view, where we have dropped $E$ to simplify notation. We distinguish between two cases. If $1+h \geqslant y+\ell$, the optimal welfare is $1+h$ and the approximation ratio is 
\begin{align*}
\frac{1+h}{\frac{2}{3}(1+h)+\frac{1}{3}(y+\ell)} &\leqslant \frac{3}{2}
\end{align*}
since $y+\ell\geqslant 0$. If $1+h\leqslant y+\ell$, the optimal welfare is $y+\ell$ and the approximation ratio is
\begin{align*}
\frac{y+\ell}{\frac{2}{3}(1+h)+\frac{1}{3}(y+\ell)}&= \frac{1}{\frac{2}{3}\frac{1+h}{y+\ell}+\frac{1}{3}} \leqslant \frac{3}{2}
\end{align*}
since $\frac{1+h}{y+\ell}\geqslant \frac{1}{2}$.
\end{proof}

\begin{theorem}\label{thm:ordinal-lower}
The approximation ratio of any ordinal mechanism is at least $3/2$.
\end{theorem}

\begin{proof}
Let $\epsilon \in (0,1/2)$ and consider the following two profiles:
$$\left(\begin{array}{c c c} 1 & \epsilon & 0\\ 0 & \epsilon & 1 \end{array}\right) \mbox{ and } 
\left(\begin{array}{c c c} 1 & 1-\epsilon & 0\\ 0 & 1-\epsilon & 1\end{array}\right).$$ 
Since the order of the expert valuations and the bids is the same in both profiles, an ordinal mechanism behaves identically in all these profiles for every $\epsilon\in (0,1/2)$. Assume that such a mechanism selects the middle option with probability $p$. Then, the approximation ratio of this mechanism is at least the maximum between its approximation ratio for these two profiles. Considering all profiles for $\epsilon\in (0,1/2)$, we get an approximation ratio of at least 
\begin{align*}
\sup_{\epsilon\in (0,1/2)}\left\{ \frac{1}{1-p+2\epsilon p}, \frac{2(1-\epsilon)}{1-p+2(1-\epsilon)p} \right\}=\max\left\{\frac{1}{1-p}, \frac{2}{1+p} \right\} .
\end{align*}
This is minimized to $3/2$ for $p=1/3$.
\end{proof}


\section{Bid-independent mechanisms}\label{sec:bim}
In this section, we consider cardinal mechanisms, but restrict our attention to ones that ignore the bids and base their decisions only on the expert's report. It is convenient to use the expert's view of profiles $\left(\begin{array}{c c c} 1 & x& 0\\ h_a & \ell_a & z_a\end{array}\right)$. Then, a {\em bid-independent} mechanism can be thought of as using univariate functions $g^M$, $f^M$, and $\eta^M$ which indicate the probability of selecting the expert's first, second, and third favorite option when she has value $x\in [0,1]$ for the second favorite option. We drop $a$ and $M$ from notation since the profile view and the mechanism will be clear from context. The next lemma provides sufficient and necessary conditions for bid-independent mechanisms with good approximation ratio.

\begin{lemma}\label{lem:bim-conditions}
Let $M$ be a bid-independent mechanism that uses functions $g$, $f$ and $\eta$. Then $M$ has approximation ratio at most $\rho$ if and only if the inequalities
\begin{align}\label[ineq]{eq:lem-bim-1}
2g(x)+xf(x)& \geqslant 2/\rho\\\label[ineq]{eq:lem-bim-2}
g(x)+(1+x)f(x) &\geqslant (1+x)/\rho
\end{align}
hold for every $x\in [0,1]$.
\end{lemma}

\begin{proof}
Consider the application of $M$ on the profile $\left(\begin{array}{c c c} 1 & x& 0\\ h & \ell & z\end{array}\right)$. If $1+h\geqslant x+\ell$ the optimal welfare is $1+h$ and the approximation ratio is 
\begin{align*}
\frac{1+h}{(1+h)g(x)+(x+\ell)f(x)+z\eta(x)} \leqslant \frac{1+h}{(1+h)g(x)+(x+\ell)f(x)}
\leqslant \frac{2}{2g(x)+xf(x)}.
\end{align*}
The first inequality follows since $z\eta(x)\geqslant 0$, and the second inequality follows since the expression in the middle is non-increasing in $\ell \geq 0$ and non-decreasing in $h \leq 1$. Then, the first inequality of the statement follows as a sufficient condition so that $M$ has approximation ratio at most $\rho$. To see why it is also necessary, observe that the inequalities in the derivation above are tight for the profile with $h=1$, $\ell=0$, and $z=0$.

If $1+h\leqslant x+\ell$ the optimal welfare is $x+\ell$ and the approximation ratio is 
\begin{align*}
\frac{x+\ell}{(1+h)g(x)+(x+\ell)f(x)+z\eta(x)} \leqslant \frac{x+\ell}{(1+h)g(x)+(x+\ell)f(x)}
\leqslant \frac{1+x}{g(x)+(1+x)f(x)}.
\end{align*}
The first inequality follows since again $z\eta(x)\geqslant 0$, while now the second inequality follows since the expression in the middle is non-decreasing in $\ell \leq 1$ and non-increasing in $h \geq 0$. Then, the second inequality of the statement follows as a sufficient condition so that $M$ has approximation ratio at most $\rho$. To see why it is also necessary, observe that the two inequalities in the derivation above are tight for the profile with $h=0$, $\ell=1$, and $z=0$.
\end{proof}

Truthfulness of bid-independent mechanisms in terms of the agents follows trivially (since the bids are ignored). In order to guarantee truthfulness from the expert's side, we will use the characterization of ECh-IC from \cref{lem:ECh-IC} together with additional conditions that will guarantee ESw-IC. These are provided by the next lemma.

\begin{lemma}\label{lem:ESw-IC}
An ECh-IC bid-independent mechanism is truthful if and only if the functions $g$, $f$, and $\eta$ it uses satisfy $g(x)\geqslant f(x')$ and $f(x)\geqslant \eta(x')$ for every pair $x,x'\in (0,1)$.
\end{lemma}

\begin{proof}
We first show that the first condition is necessary. Assume that the first condition is violated, i.e., $f(x_1)>g(x_2)$ for two points $x_1,x_2\in (0,1)$. If $x_1>x_2$, by the monotonicity of $g$ (due to ECh-IC; see Lemma~\ref{lem:ECh-IC}) we have $g(x_1)\leqslant g(x_2)$ and $f(x_1)>g(x_1)$. Otherwise, by the monotonicity of $f$, we have $f(x_2)\geqslant f(x_1)$ and $f(x_2)>g(x_2)$. In any case, there must exist $x^*\in (0,1)$ such that $f(x^*)>g(x^*)$. Now consider the swap from expert valuation profile $(1,x^*,0)$ to the profile $(x^*,1,0)$. The utility of the expert in the initial true profile is $g(x^*)+x^*f(x^*)$ while her utility at the new profile becomes $f(x^*)+x^*g(x^*)$, which is strictly higher. 

Now, we show that the second condition is necessary. Again, assuming that the second condition is violated, we obtain that there is a point $x^*\in (0,1)$ such that $\eta(x^*)>f(x^*)$. Now, the swap from expert's valuation profile $(1,x^*,0)$ to the profile $(1,0,x^*)$ increases the utility of the expert from $g(x^*)+x^*f(x^*)$ to $g(x^*)+x^*\eta(x^*)$, which is again strictly higher.

In order to show that the condition is sufficient for ECh-IC, we need to consider five possible attempts for valuation swap by the expert. 
\medskip

\noindent{\em Case 1.} Consider the swap from the valuation profile $(1, x, 0)$ to the profile $(1, 0, x')$. The utility of the expert at the new profile is $g(x')+x\eta(x')\leqslant g(0)+\int_0^x{f(t)\ud t} = g(x)+xf(x)$, where the inequality holds due to the fact that $\eta(x')\leqslant f(t)$, for every $t\in [0,x]$. Observe that the RHS of the derivation is the expert's utility at the initial true profile. 
\medskip

\noindent{\em Case 2.} Consider the swap from the valuation profile $(1,x,0)$ to the profile $(x',1,0)$. The utility of the expert at the new profile is $f(x')+xg(x')\leqslant g(x')+xf(x')=g(x')+x'f(x')+(x-x')f(x')\leqslant g(x)+xf(x)$, which is her utility at the initial true profile. The first inequality follows by the condition $g(x')\geqslant f(x)$ of the lemma and the second one is due to the convexity of function $g(x)+xf(x)$. See also the proof of \cref{lem:ECh-IC}.
\medskip

\noindent{\em Case 3.} Consider the swap from the valuation profile $(1,x,0)$ to the profile $(x',0,1)$. The utility of the expert at the new profile is $f(x')+x\eta(x')$, which is at most $g(x)+xf(x)$ due to the conditions of the lemma. 
\medskip

\noindent{\em Case 4.} Consider the swap from the valuation profile $(1,x,0)$ to the profile $(0,x',1)$. The utility of the expert at the new profile is $\eta(x')+xf(x')\leqslant f(x)+xg(x)\leqslant g(x)+xf(x)$, which is her utility at the initial true profile. 
\medskip

\noindent{\em Case 5.} Consider the swap from the valuation profile $(1,x,0)$ to the profile $(0,1,x')$. The utility of the expert at the new profile is $\eta(x')+xg(x')\leqslant f(x')+xg(x')$ and the proof proceeds as in Case 2 above.
\end{proof}

We are now ready to propose our mechanism BIM. Let $\tau=-W\left(-\frac{1}{2e}\right)$, where $W$ is the Lambert function, i.e., $\tau$ is the solution of the equation $2\tau = e^{\tau-1}$. Mechanism BIM is defined as follows:
\begin{minipage}[h!]{0.33\textwidth}
\begin{align*}
g(x) = 
\begin{cases}
\frac{1+\tau}{1+3\tau}, & \!\!\!\! x \in [0,\tau]\\
\frac{2\tau(1+x) e^{1-x}}{1+3\tau} , & \!\!\!\! x \in [\tau,1]
\end{cases}
\end{align*}\end{minipage}
\begin{minipage}[h!]{0.33\textwidth}
\begin{align*}
f(x) = 
\begin{cases}
\frac{\tau}{1+3\tau}, & \!\!\!\! x \in [0,\tau] \\
\frac{1+\tau -2\tau e^{1-x}}{1+3\tau} , & \!\!\!\! x \in [\tau,1]
\end{cases}
\end{align*}
\end{minipage}
\begin{minipage}[h!]{0.33\textwidth}
\begin{align*}
\eta(x) = 
\begin{cases}
\frac{\tau}{1+3\tau}, & \!\!\!\! x \in [0,\tau] \\
\frac{2\tau(1-xe^{1-x})}{1+3\tau} , & \!\!\!\! x \in [\tau,1]
\end{cases}
\end{align*}
\end{minipage}

\vspace{10pt}
\noindent BIM is depicted in \cref{fig:bim}. All functions are constant in $[0,\tau]$ and have (admittedly, counter-intuitive at first glance) exponential terms in $[\tau,1]$. Interestingly, 
BIM is the unique best possible solution to a set of constraints that need to be satisfied by all bid-independent truthful mechanisms, which are derived in the proof of
\cref{thm:bim-lower}. 
Its properties are proved in the next statement.

\begin{figure}[t]
\center\includegraphics[scale=0.65]{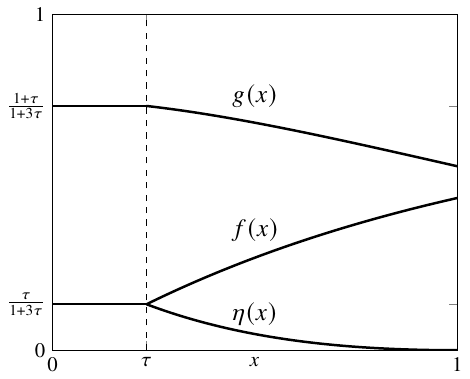}
\caption{A pictorial view of the lottery used by mechanism BIM. $\tau$ is the solution of the equation $2\tau=e^{\tau-1}$.}
\label{fig:bim}
\end{figure}

\begin{theorem}\label{thm:bim}
Mechanism BIM is truthful and has approximation ratio at most $\frac{1-3W\left(-\frac{1}{2e} \right)}{1-W\left(-\frac{1}{2e} \right)} \approx 1.37657$, where $W$ is the Lambert function.
\end{theorem}

\begin{proof}
Tedious calculations can verify that BIM is truthful. The function $f$ is non-decreasing in $x$ and $g$ is defined exactly as in \cref{eq:lem:ECh-IC}. Hence, ECh-IC follows by \cref{lem:ECh-IC}. ESw-IC follows since $f$, $g$, and $h$ satisfy the conditions of \cref{lem:ESw-IC}.

Now, let $\rho=\frac{1+3\tau}{1+\tau}$. We use the definition of BIM and \cref{lem:bim-conditions} to show the bound on the approximation ratio. If $x\in [0,\tau]$, \cref{eq:lem-bim-1} and \cref{eq:lem-bim-2} are clearly satisfied since $x\geqslant 0$ and $x\leqslant \tau$, respectively. If $x\in [\tau,1]$, we have 
\begin{align*}
2g(x)+xf(x) &= 2 \frac{2\tau(1+x) e^{1-x}}{1+3\tau} + x \frac{1+\tau -2\tau e^{1-x}}{1+3\tau},
\end{align*}
which is minimized for $x=\tau$ (recall that $2\tau=e^{\tau-1}$) at $\frac{2+2\tau+\tau^2}{1+3\tau}\geqslant 2/\rho$.  Hence, \cref{eq:lem-bim-1} holds. Also, \cref{eq:lem-bim-2} can be easily seen to hold with equality.
\end{proof}

We now show that the above mechanism is optimal among all bid-independent truthful mechanisms. The proof exploits the characterization of ECh-IC mechanisms from \cref{lem:ECh-IC}, the characterization of ESw-IC bid-independent mechanisms from \cref{lem:ESw-IC}, and \cref{lem:bim-conditions}.

\begin{theorem}\label{thm:bim-lower}
The approximation ratio of any truthful bid-independent mechanism is at least $\frac{1-3W\left(-\frac{1}{2e}\right)}{1-W\left(-\frac{1}{2e}\right)} \approx 1.37657$, where $W$ is the Lambert function.
\end{theorem}

\begin{proof}
Let $M$ be a bid-independent mechanism that uses functions $g$, $f$, and $h$ to define the probability of selecting the expert's first, second, and third favorite option and has approximation ratio $\rho\geqslant 1$. Let $\alpha$ be any value in $[0,1]$.
	
By the necessary condition \cref{eq:lem:ECh-IC} for ECh-IC in \cref{lem:ECh-IC}, we know that 
\begin{align}\label{eq:mmm}
g(x) &=g(0)-xf(x)+\int_0^x{f(t)\ud t}.
\end{align}
Due to the fact that $f(1)+g(1)\leqslant 1$, we have
\begin{align}\label[ineq]{eq:lb-ineq1}
g(0)+\int_0^1{f(t)\ud t} &\leqslant 1.
\end{align}
By the necessary condition for ESw-IC in \cref{lem:ESw-IC} and since $g$ is non-increasing (by \cref{lem:ECh-IC}), we also have $f(x)\geqslant \eta(x)=1-f(x)-g(x)\geqslant 1-f(x)-g(0)$, i.e., $g(0)+2f(x)\geqslant 1$, for $x\in (0,1)$. Integrating in the interval $(0,\alpha]$, we get
\begin{align}\label[ineq]{eq:lb-ineq2}
\alpha g(0) + 2 \int_0^\alpha{f(t)\ud t} &\geqslant \alpha.
\end{align}
Since, the mechanism is $\rho$-approximate, \cref{lem:bim-conditions} yields
\begin{align}\label[ineq]{eq:lb-ineq3}
g(0) &\geqslant 1/\rho
\end{align}
(by applying \cref{eq:lem-bim-1} with $x=0$) and
\begin{align*}
g(x)+(1+x)f(x) &\geqslant (1+x)/\rho, \forall x\in [\alpha,1].
\end{align*}
Using \cref{eq:mmm}, this last inequality becomes
\begin{align*}
g(0)+f(x)+\int_0^x{f(t)\ud t} &\geqslant (1+x)/\rho, \forall x \in [\alpha,1].
\end{align*}
Now, let $\lambda$ be a continuous function with $\lambda(x)\leqslant f(x)$ in $[\alpha,1]$ such that
\begin{align*}
g(0)+\int_0^\alpha{f(t)\ud t}+ \int_\alpha^x{\lambda(t)dt}+\lambda(x)=(1+x)/\rho.
\end{align*}
Setting $\Lambda(x)=\int_\alpha^x{\lambda(t)dt}$ (clearly, $\Lambda$ is differentiable due to the continuity of $\lambda$ in $[0,1]$), we get the differential equation 
\begin{align*}
g(0)+\int_0^\alpha{f(t)\ud t}+ \Lambda(x)+\Lambda'(x)&=(1+x)/\rho
\end{align*}
which, given that $\Lambda(\alpha)=0$, has the solution
\begin{align*}
\Lambda(x) &= \frac{x}{\rho}-g(0)-\int_0^\alpha{f(t)\ud t}+\left(g(0)-\frac{\alpha}{\rho}+\int_0^\alpha{f(t)\ud t}\right) \exp{(\alpha-x)}
\end{align*}
for $x\in [\alpha,1]$. Hence, 
\begin{align}\label[ineq]{eq:lb-ineq4}
\int_\alpha^1{f(t)\ud t} &\geqslant \Lambda(1) = \frac{1-\alpha e^{\alpha-1}}{\rho}-\left(1-e^{\alpha-1}\right)g(0)-\left(1-e^{\alpha-1}\right)\int_0^\alpha{f(t)\ud t}.
\end{align}
	
Now, by multiplying \cref{eq:lb-ineq1}, \cref{eq:lb-ineq2}, \cref{eq:lb-ineq3}, and \cref{eq:lb-ineq4} by coefficients $2$, $e^{\alpha-1}$, $(2-\alpha)e^{\alpha-1}$, and $2$, respectively, and then summing them, we obtain
\begin{align*}
\rho &\geqslant \frac{2-3\alpha e^{\alpha-1}+2e^{\alpha-1}}{2-\alpha e^{\alpha-1}}.
\end{align*}
Picking $\alpha = -W\left(-\frac{1}{2e}\right)$ (i.e., $\alpha$ satisfies $e^{\alpha-1}=2\alpha$), we get that
\begin{align*}
\rho &\geqslant \frac{2-6\alpha^2+4\alpha}{2-2\alpha^2} = \frac{1+3\alpha}{1+\alpha} =\frac{1-3W\left(-\frac{1}{2e}\right)}{1-W\left(-\frac{1}{2e}\right)}.
\end{align*}
This completes the proof.
\end{proof}


\section{Expert-independent mechanisms}\label{sec:eim}
In this section, we consider mechanisms that depend only on the bids. Now, it is convenient to use the agents' view of profiles $\left[\begin{array}{c c c}h_E & \ell_E & n_E\\1 & y& 0\end{array}\right]$. Then, an expert-independent mechanism can be thought of as using univariate functions $d^M$, $c^M$, and $e^M$ which indicate the probability of selecting the high-bidder, the low-bidder, and the option $\oslash$ in terms of the normalized low-bid $y$. Again, we drop $E$ and $M$ from notation. Following the same roadmap as in the previous section, the next lemma provides sufficient and necessary conditions for expert-independent mechanisms with good approximation ratio.

\begin{lemma}\label{lem:expert-independent-worst-case}
Let $M$ be an expert-independent mechanism that uses functions $d$, $c$, and $e$ with $d(y)=1-c(y)$ and $e(y)=0$ for $y\in [0,1]$. If 
\begin{align}\label[cond]{eq:condition-ei}
\frac{1}{\rho}-\frac{1-1/\rho}{y} \leqslant c(y) \leqslant \frac{2(1-1/\rho)}{2-y}
\end{align}
for every $y\in [0,1]$, then $M$ has approximation ratio at most $\rho$. \cref{eq:condition-ei} is necessary for every $\rho$-approximate expert-independent mechanism.
\end{lemma}

\begin{proof}
Consider the application of $M$ on the profile with agents' view $\left[\begin{array}{c c c} h & \ell & n\\1 &y & 0\end{array}\right]$. We distinguish between two cases. If $1+h\geqslant y+\ell$, assuming that \cref{eq:condition-ei} is true, the approximation ratio of $M$ is
\begin{align*}
\frac{1+h}{(y+\ell)c(y)+(1+h)(1-c(y))} = \frac{1}{\frac{y+\ell}{1+h}c(y)+1-c(y)}
\leqslant \frac{1}{1-(1-y/2)c(y)} \leqslant \rho.
\end{align*}
The first inequality follows since $\frac{y+\ell}{1+h}\geqslant y/2$ when $y\in [0,1]$, while the second one is essentially the right inequality of \cref{eq:condition-ei}. 

Otherwise, if $1+h\leqslant y+\ell$, again assuming that \cref{eq:condition-ei} is true, the approximation ratio of $M$ is
\begin{align*}
\frac{y+\ell}{(y+\ell)c(y)+(1+h)(1-c(y))} = \frac{1}{c(y)+\frac{1+h}{y+\ell}(1-c(y))} 
\leqslant \frac{1+y}{1+yc(y)} \leqslant \rho.
\end{align*}
The first inequality follows since $\frac{1+h}{y+\ell}\geqslant \frac{1}{1+y}$ when $y\in [0,1]$, while the second one is now the left inequality of  \cref{eq:condition-ei}. 

To see that \cref{eq:condition-ei} is necessary for every mechanism, first consider a mechanism $M'$ that uses functions $\overline{d}$, $\overline{c}$ and $\overline{e}$ such that the function $\overline{c}$ violates the left inequality of \cref{eq:condition-ei}, i.e., $\overline{c}(y^*)<\frac{1}{\rho}-\frac{1-1/\rho}{y^*}$ for some $y^*\in [0,1]$. Then, using this inequality and the fact that $\overline{d}(y^*)\leqslant 1-\overline{c}(y^*)$, the approximation ratio of $M'$ for the profile 
$\left[\begin{array}{c c c}0 & 1& 0\\1&y^*&0\end{array}\right]$ 
is
\begin{align*}
\frac{y^*+1}{(y^*+1)\overline{c}(y^*)+\overline{d}(y^*)} &\geqslant \frac{1+y^*}{1+y^*\overline{c}(y^*)} > \rho.
\end{align*}
Now, assume that $\overline{c}$ violates the right inequality in \cref{eq:condition-ei}, i.e., $\overline{c}(y^*)>\frac{2(1-1/\rho)}{2-y^*}$ for some $y^*\in [0,1]$. Then, using this together with the fact that $\overline{d}(y^*)\leqslant 1-\overline{c}(y^*)$, the approximation ratio of $M'$ for the profile 
$\left[\begin{array}{c c c}1 & 0& 0\\1&y^*&0\end{array}\right]$ 
is 
\begin{align*}
\frac{2}{2\overline{d}(y^*)+y^*\overline{c}(y^*)} &\geqslant \frac{2}{2-(2-y^*)\overline{c}(y^*)}> \rho
\end{align*}
as desired.
\end{proof}

\begin{figure}[t]
\centering
\includegraphics[scale=0.65]{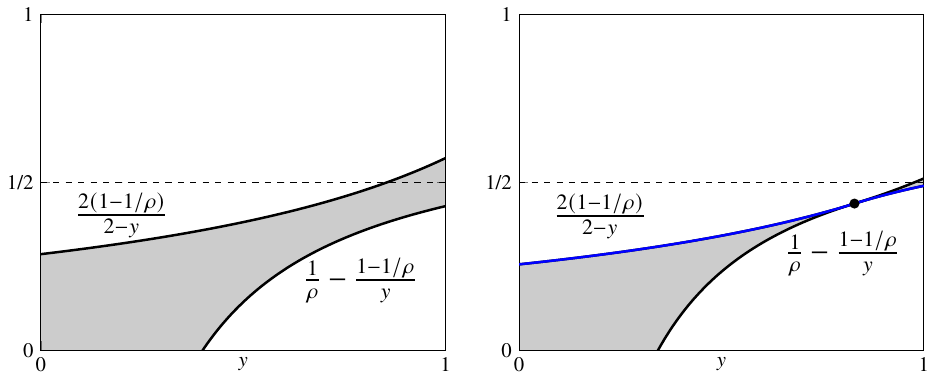}
\caption{Pictorial views of the statement of \cref{lem:expert-independent-worst-case} for $\rho=7/5$ (left) and $\rho=7-4\sqrt{2}$ (right). The gray area corresponds to the available space that allows for the definition of (the function $c$ of) expert-independent mechanisms with approximation ratio at most $\rho$. The truthfulness requirements are obtained if $c$ is non-decreasing (for BCh-IC; see Lemma~\ref{lem:BCh-IC}) and takes values not higher than $1/2$ (for BSw-IC; see Lemma~\ref{lem:BSw-IC}). The blue line at the right is the $c$-function used by our mechanism EIM, which has all the above properties for $\rho=7-4\sqrt{2}$.}
\label{fig:eim}
\end{figure}

\cref{fig:eim} shows the available space (gray area) for the definition of the function $c(y)$, so that the corresponding mechanism has an approximation ratio of at most $\rho$. It can be easily verified that the value $\rho=7-4\sqrt{2}\approx 1.3431$ (see the right part of \cref{fig:eim}) is the minimum value for which the LHS of \cref{eq:condition-ei} in \cref{lem:expert-independent-worst-case} is smaller than or equal to the RHS so that a function satisfying \cref{eq:condition-ei} does exist.

Our aim now is to define an expert-independent truthful mechanism achieving the best possible approximation ratio of $\rho=7-4\sqrt{2}$. Clearly, truthfulness for the expert follows trivially (since the expert's report is ignored). We restrict our attention to the design of a mechanism that never selects option $\oslash$, and thus $d(y)=1-c(y)$ for every $y\in [0,1]$. \cref{lem:BCh-IC} and \cref{lem:expert-independent-worst-case} guide this design as follows. In order to be BCh-IC and $\rho$-approximate, our mechanism should use a non-decreasing function $c(y)$ in the space available by \cref{eq:condition-ei}. Still, we need to guarantee BSw-IC; the next lemma gives us the additional sufficient (and necessary) condition.

\begin{lemma}\label{lem:BSw-IC}
A BCh-IC expert-independent mechanism is truthful if and only if $d(1)\geqslant c(1)$. 
\end{lemma}

\begin{proof}
Consider an attempted bid swap according to which the low-bidder increases her normalized bid of $y$ so that it becomes the high-bidder and the normalized bid of the other agent is $y'$. Essentially, this attempted bid swap modifies the initial profile $\left[\begin{array}{c c c}h & \ell& n\\1 & y& 0\end{array}\right]$ to $\left[\begin{array}{c c c} \ell & h & n\\1 & y' & 0\end{array}\right]$. The deviating agent corresponds to the middle column in the initial profile and has probability $c(y)$ of being selected. In the new profile, she corresponds to the first column, and has probability $d(y')$ of being selected. So, the necessary and sufficient condition so that BSw-IC is guaranteed is $c(y)\leqslant d(y')$ for every $y,y'\in [0,1]$. Since, by \cref{lem:BCh-IC}, $c$ and $d$ are non-decreasing and non-increasing, respectively, this condition boils down to $d(1)\geqslant c(1)$.
	
The case in which the high-bidder decreases her bid so that it gets a normalized value of $y'$ is symmetric.
\end{proof}

\noindent For mechanisms with $d(y)=1-c(y)$ for $y\in [0,1]$, the condition of Lemma~\ref{lem:BSw-IC} becomes $c(1)\leq 1/2$.

We are ready to propose our mechanism EIM, which uses the following functions. For $\rho=7-4\sqrt{2}$,
\begin{align*}
c(y) &= 
\begin{cases}
\frac{2(1-1/\rho)}{2-y}, & y\in [0,\frac{3-\rho}{2}]\\
\frac{1}{\rho}-\frac{1-1/\rho}{y}, & y\in [\frac{3-\rho}{2},1]
\end{cases}
\end{align*}
and $d(y)=1-c(y)$ for $y\in [0,1]$.

Essentially, EIM uses the blue line in \cref{fig:eim}, which consists of the curve that upper-bounds the gray area up to point $\frac{3-\rho}{2} = 2\sqrt{2}-2$ and the curve that lower-bounds the gray area after that point. The properties of mechanism EIM are summarized in the next statement. It should be clear though that the statement holds for every mechanism that uses a non-decreasing function in the gray area that is below $1/2$. Given the discussion about the optimality of $\rho=7-4\sqrt{2}$ above, all these mechanisms are optimal within the class of expert-independent mechanisms.

\begin{theorem}\label{thm:eim}
Mechanism EIM is truthful and has approximation ratio at most $7-4\sqrt{2} \approx 1.3431$. This ratio is optimal among all truthful expert-independent mechanisms.
\end{theorem}


\section{Beyond expert-independent mechanisms}\label{sec:beyond}
In this section, we present a template for the design of even better truthful mechanisms, compared to those presented in the previous sections. The template strengthens expert-independent mechanisms by exploiting a single additional bit of information that allows to distinguish between profiles that have the same (normalized) bid values.

We denote by $\calT$ the set of mechanisms that are produced according to our template. In order to define a mechanism $M\in \calT$, it is convenient to use the agents' view of a profile as $\left[\begin{array}{c c c} h & \ell & n\\ 1 & y & 0\end{array}\right]$. We partition the profiles of $\mathcal{D}$ into two categories. Category $T1$ contains all profiles with $\ell>h$ or with $\ell=h$ such that the tie between the expert valuations $\ell$ and $h$ is resolved in favor of the low-bidder. All other profiles belong to category $T2$.

For each profile in category $T1$, mechanism $M$ selects the low-bidder with probability $c(y,T1)$,which is non-decreasing in $y$, and the high-bidder with probability $1-c(y,T1)$. For each profile in category $T2$, mechanism $M$ selects the low-bidder with probability $0$, and the high-bidder with probability $1$. Different mechanisms following our template can be defined using different functions $c(y,T1)$. The mechanisms of the template ignore neither the bids nor the expert's report; still, we can show that they are truthful.

\begin{lemma}\label{lem:template-truthfulness}
Every mechanism $M\in \calT$ is truthful.
\end{lemma}

\begin{proof}
We first show that $M$ is truthful for the agents. BCh-IC follows easily by \cref{lem:BCh-IC}, since $c(y,T1)$ and $c(y,T2)$ are non-decreasing in $y$. To show BSw-IC, notice that a bid swap attempt from a profile of category $T1$ creates a profile of category $T2$ and vice versa. This involves either the high-bidder who decreases her bid to become the low-bidder in the new profile, or the low-bidder who increases her bid to become the high-bidder in the new profile. In both cases, the increase or decrease in the selection probability according to $M$ follows the increase or decrease of the deviating bid.

To show that $M$ is truthful for the expert, first observe that according to the expert's view, the lottery uses constant functions $f$, $g$, and $h$ in terms of her value for her second favorite option. Hence, \cref{lem:ECh-IC} implies ECh-IC. To show ESw-IC, observe again that an expert's report swap attempt from a profile of category $T1$ creates a profile of category $T2$ and vice versa. The expected utility that $M$ yields to the expert in the initial profile $\left[\begin{array}{c c c} h & \ell & n\\ 1 & y & 0\end{array}\right]$ is $\ell c(y,T1)+h(1-c(y,T1))=h+(\ell-h)c(y,T1)\geqslant h$ if it is of category $T1$ and $h+(\ell-h)c(y,T2)=h$ if it is of category $T2$. After the deviation, the utility of the expert becomes $\ell c(y,T1)+h(1-c(y,T1))=h+(\ell-h)c(y,T1)\leqslant h$ if the new profile is of category $T1$ and $h+(\ell-h)c(y,T2)=h$ if it is of category $T2$. Hence, such a swap attempt is never profitable for the expert.
\end{proof}

The next lemma is useful in proving bounds on the approximation ratio of mechanisms in $\calT$.

\begin{lemma}\label{lem:template-approx}
Let $M$ be a mechanism of $\calT$ and $\rho\geqslant 1$ be such that the function $c(y,T1)$ used by $M$ satisfies 
$$\frac{1}{\rho}-\frac{1-1/\rho}{y} \leqslant c(y,T1) \leqslant \frac{1-1/\rho}{1-y}.$$
Then, $M$ has approximation ratio at most $\rho$.
\end{lemma}

\begin{proof}
Clearly, the approximation ratio of $M$ in profiles of category $T2$ is always $1$ since the mechanism takes the optimal decision of selecting the high-bidder with probability $1$.

Now, consider a profile $\left[\begin{array}{c c c} h & \ell & n \\ 1 & y & 0\end{array}\right]$ of category $T1$, i.e., $\ell\geqslant h$. We distinguish between two cases. If $1+h\geqslant y+\ell$, then the approximation ratio of $M$ is
\begin{align*}
\frac{1+h}{(y+\ell)c(y,T1)+(1+h)(1-c(y,T1))} &= \frac{1}{1-c(y,T1)+\frac{y+\ell}{1+h}c(y,T1)}
\leqslant \frac{1}{1-(1-y)c(y,T1)} \leqslant \rho.
\end{align*}
The first inequality follows since $\frac{y+\ell}{1+h}\geqslant y$ when $y\in [0,1]$ and $\ell\geqslant h\geqslant 0$, while the second one is due to the right inequality in the condition of the lemma.

Otherwise, if $1+h\leqslant y+\ell$, the approximation ratio of $M$ is
\begin{align*}
\frac{y+\ell}{(y+\ell)c(y,T1)+(1+h)(1-c(y,T1))} &= \frac{1}{c(y,T1)+\frac{1+h}{y+\ell}(1-c(y,T1))}
\leqslant \frac{1+y}{1+yc(y,T1)} \leqslant \rho.
\end{align*}
The first inequality follows since $\frac{1+h}{y+\ell}\geqslant \frac{1}{1+y}$ when $y\in [0,1]$ and $h\geqslant \ell\geqslant 0$; the second one is due to the left inequality in the condition of the lemma.
\end{proof}

\begin{figure}[t]
\centering
\includegraphics[scale=0.65]{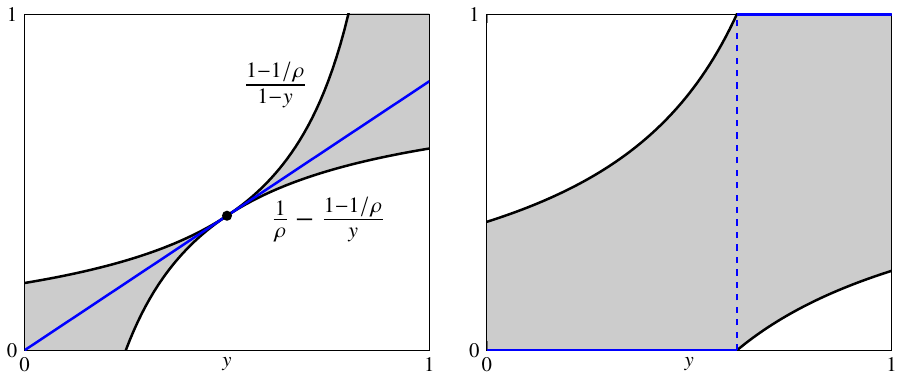}
\caption{Pictorial views of the statement in \cref{lem:template-approx} for $\rho=5/4$ (left) and $\rho=\phi$ (right). The gray area in both plots corresponds to the available space for the definition of $\calT$ mechanisms (e.g., mechanisms $R$ and $D$ which are depicted by the blue lines).}
\label{fig:template}
\end{figure}

The conditions in the statement of \cref{lem:template-approx} are depicted in the left (for $\rho=5/4$) and right plot (for $\rho=\phi$) of \cref{fig:template}. The gray area represents the available space for the definition of (the non-decreasing) function $c(y,T1)$ that a mechanism of $\calT$ should use on profiles of category $T1$ so that its approximation ratio is at most $\rho$.

These plots explain the definition of the next two mechanisms that follow our template: the randomized mechanism $R$ and the deterministic mechanism $D$. For each profile of category $T1$, mechanisms $R$ and $D$ use the functions 

\begin{center}
\begin{minipage}[h!]{0.4\textwidth}
\begin{align*}
c^R(y,T1) & = \frac{4y}{5}
\end{align*}
\end{minipage}
\begin{minipage}[h!]{0.44\textwidth}
\begin{align*}
c^D(y,T1) & = 
\begin{cases}
0, & y\in [0,1/\phi)\\
1, & y\in [1/\phi,1]
\end{cases}
\end{align*}
\end{minipage}
\end{center}

\noindent
corresponding to the blue lines in the left and right plots of \cref{fig:template}, respectively; $\phi=\frac{1+\sqrt{5}}{2}\approx 1.618$ is the golden ratio. Their properties are as follows.

\begin{theorem}\label{thm:R}
Mechanisms $R$ and $D$ are $5/4$- and $\phi$-approximate truthful mechanisms, respectively.
\end{theorem}

\begin{proof}
Since $R,D\in \calT$, their truthfulness follows by \cref{lem:template-truthfulness}. The approximation ratios follow by verifying that the conditions of \cref{lem:template-approx} are satisfied for $\rho=5/4$ and $\rho=\phi$, respectively.
\end{proof}

We remark that the condition of \cref{lem:template-approx} can be proved to be not only sufficient but also necessary for achieving a $\rho$-approximation (using mechanisms from $\calT$). Then, it can be easily seen that the value of $5/4$ is the lowest value for which the condition of the lemma is feasible. Hence, mechanism $R$ is best possible among mechanisms that use our template. More interestingly, $5/4$ turns out to be the lower bound of any mechanism that \emph{always sells}, as we prove in the next theorem. Mechanism $D$ will be proved to be optimal among all deterministic truthful mechanisms in the next section.

\begin{theorem}\label{thm:always-sell}
The approximation ratio of any mechanism that always sells is at least $5/4$.
\end{theorem}

\begin{proof}
Consider profiles in agents' view $\left[\begin{array}{c c c}h & \ell& n\\1 & y & 0\end{array}\right]$ and let $M$ be any truthful always-sell mechanism. Then, $M$ can be thought of as using functions $d(y,h,\ell,n)$, $c(y,h,\ell,n)$ and $e(y,h,\ell,n)$ to assign probabilities to the high-bidder, the low-bidder and the no-sale option, respectively, such that $d(y,h,\ell,n)=1-c(y,h,\ell,n)$ and $e(y,h,\ell,n)=0$.

Since $M$ is truthful for the expert, the expert does not have any incentive to misreport her valuations from $(h,\ell,n)$ to $(h',\ell',n')$, for any $\ell>h$ and $\ell'>h'$. This means that
\begin{align*}
h \cdot (1-c(y,h,\ell,n)) + \ell \cdot c(y,h,\ell,n) \geq h \cdot (1-c(y,h',\ell',n')) + \ell \cdot c(y,h',\ell',n')
\end{align*}
or, equivalently, since $\ell > h$, 
\begin{align}\label[ineq]{eq:first-side}
c(y,h,\ell,n) \geq c(y,h',\ell',n')
\end{align}
Similarly, the expert does not have incentive to misreport her valuations from $(h',\ell',n')$ to $(h,\ell,n)$, for any $\ell>h$ and $\ell'>h'$. This gives us that
\begin{align*}
h' \cdot (1-c(y,h',\ell',n')) + \ell' \cdot c(y,h',\ell',n') \geq h' \cdot (1-c(y,h,\ell,n)) + \ell' \cdot c(y,h,\ell,n)
\end{align*}
or, equivalently, since $\ell' > h'$, 
\begin{align}\label[ineq]{eq:second-side}
c(y,h',\ell',n') \geq c(y,h,\ell,n)
\end{align}
Therefore, by \cref{eq:first-side} and \cref{eq:second-side}, we have that $c(y,h,\ell,n)$ is constant in all profiles $\left[\begin{array}{c c c}h & \ell& n\\1 & y & 0\end{array}\right]$ with $\ell>h$.
		
Now, let $\epsilon \in (0,1/2)$ and consider the following two profiles:
$$\left[\begin{array}{c c c}0 & 1 & 0\\1 & 1/2 & 0\end{array}\right] \mbox{ and }\left[\begin{array}{c c c}0 & \epsilon& 1\\1 & 1/2 & 0\end{array}\right]$$
Since $\ell > h$ in both profiles, any truthful mechanism $M$ that always sells the item behaves identically in all such profiles, for any $\epsilon \in (0,1/2)$. Assume that such a mechanism $M$ selects the low-bidder with probability $p$ (and the high-bidder with probability $1-p$). Then, the approximation ratio of $M$ is at least the maximum between its approximation ratio for these profiles, i.e., 
\begin{align*}
\sup_{\epsilon\in (0,1/2)}\left\{ \frac{\frac{3}{2}}{1-p+\frac{3}{2} p}, \frac{1}{1-p+(\epsilon+\frac{1}{2})p} \right\}=\max\left\{\frac{3}{2+p}, \frac{2}{2-p} \right\} .
\end{align*}
This is minimized to $5/4$ for $p=2/5$.
\end{proof}


\section{Unconditional lower bounds}\label{sec:lower-bounds}
In the previous sections, we presented (or informally discussed) lower bounds on the approximation ratio of truthful mechanisms belonging to particular classes. Here, we present our most general lower bound that holds for every truthful mechanism. The proof exploits the ECh-IC characterization from \cref{lem:ECh-IC}.

\begin{theorem}\label{thm:lower}
The approximation ratio of any truthful mechanism is at least $1.14078$.
\end{theorem}

\begin{proof}
Let $\gamma \in [0,1]$ be such that $1-2\gamma-4\gamma^2-2\gamma^3=0$ and $\beta=(1+\gamma)^{-1}$, i.e., $\beta \approx 0.7709$ and $\gamma\approx 0.29716$. Consider any $\rho$-approximate truthful mechanism and the profiles
$$\left(\begin{array}{c c c} 1 & \beta & 0 \\ \gamma & 1 & 0\end{array}\right) \mbox{ and } \left(\begin{array}{c c c} 1 & 0 & 0 \\ \gamma & 1 & 0\end{array}\right).$$
Since the bids are identical in both profiles, we can assume that the functions $f$ and $g$ are univariate (depending only on the expert's second highest valuation). Since the mechanism is $\rho$-approximate in both profiles, we have
\begin{align}\label[ineq]{eq:lb-gen-ineq1}
(1+\gamma)g(\beta)+(1+\beta)f(\beta) &\geqslant \frac{1+\beta}{\rho}
\end{align}
\begin{align}\label[ineq]{eq:lb-gen-ineq2}
(1+\gamma)g(0)+f(0) & \geqslant \frac{1+\gamma}{\rho}.
\end{align}
By the condition of \cref{eq:lem:ECh-IC} in \cref{lem:ECh-IC}, $g(x)=g(0)-xf(x)+\int_0^x{f(t)\ud t}$ which, due to the fact that $f$ is non-decreasing (again by \cref{lem:ECh-IC}), yields
$\int_0^\beta{f(t)\ud t}\geqslant \beta f(0)$. Hence, 
\begin{align}\label[ineq]{eq:lb-gen-ineq3}
g(x) &\geqslant g(0)-\beta f(\beta)+\beta f(0).
\end{align}
Also, clearly,
\begin{align}\label[ineq]{eq:lb-gen-ineq4}
1 &\geqslant g(\beta)+f(\beta).
\end{align}

Now, multiplying \cref{eq:lb-gen-ineq1}, \cref{eq:lb-gen-ineq2}, \cref{eq:lb-gen-ineq3}, and \cref{eq:lb-gen-ineq4} by $\frac{\gamma}{\beta+2\beta\gamma-\gamma^2}$, $\frac{\beta-\gamma}{\beta+2\beta\gamma-\gamma^2}$, $\frac{(\beta-\gamma)(1+\gamma)}{\beta+2\beta\gamma-\gamma^2}$, and $\frac{\beta(1+\gamma)}{\beta+2\beta\gamma-\gamma^2}$, and by summing them, we get
\begin{align*}
\rho &\geqslant \frac{\beta+2\beta\gamma-\gamma^2}{\beta(1+\gamma)}.
\end{align*}
Substituting $\beta$ and $\gamma$, we obtain that $\rho\geqslant 1.14078$ as desired.
\end{proof}

Our last statement shows that mechanism $D$ in \cref{sec:beyond} is best possible among all deterministic truthful mechanisms.

\begin{theorem}\label{thm:D-lower}
No truthful deterministic mechanism has approximation ratio better than $\phi$.
\end{theorem}

\begin{proof}
Let $M$ be a deterministic truthful mechanism. Consider a profile  
$\left(\begin{array}{c c c} 1 & x & 0\\ h & \ell & z\end{array}\right)$ in expert's view,
for some combination of values for $h$, $\ell$, and $z$. We will first show that $M$ selects the same option for every value of $x\in (0,1)$. Indeed, assume otherwise; due to \cref{lem:ECh-IC}, $f$ must be non-decreasing in $x$ and, hence, $f(x_1,h,\ell,z)=0$ and $f(x_2,h,\ell,z)=1$ for two different values $x_1$ and $x_2$ in $(0,1)$ with $x_1<x_2$. Let $x_3\in (x_2,1)$, i.e., $f(x_3,h,\ell,z)=1$ due to monotonicity. The condition of \cref{eq:lem:ECh-IC} in \cref{lem:ECh-IC} requires that
\begin{align*}
g(x_3,h,\ell,z) & = g(0,h,\ell,z)- x_3 +\int_0^{x_3}{f(t,h,\ell,z)\ud t}.
\end{align*}
By our assumptions on $f$ (and due to its monotonicity), we also have that
\begin{align*}
x_3-x_2 &\leqslant \int_0^{x_3}{f(t,h,\ell,z)\ud t}
\leqslant x_3-x_1.
\end{align*}
These last two (in)equalities imply that $g(0,h,\ell,z)-g(x_3,h,\ell,z)$ lies between $x_2$ and $x_3$, i.e., it is non-integer. This contradicts the fact that $M$ is deterministic.

Now let $\epsilon>0$ be negligibly small and consider the two profiles
$$\left(\begin{array}{c c c} 1 & 1-\epsilon & 0\\ 0 & 1/\phi & 1\end{array}\right) \mbox{ and } \left(\begin{array}{c c c} 1 & \epsilon/\phi^2 & 0\\ 0 & 1/\phi & 1\end{array}\right).$$
If $M$ selects the low-bidder in both profiles, its approximation ratio at the right one is $\frac{1}{\epsilon/\phi^2+1/\phi}\geqslant \phi-\epsilon$. Otherwise, its approximation ratio at the left profile is $1+1/\phi-\epsilon$. In any case, the approximation ratio is at least $\phi-\epsilon$, and the proof is complete.
\end{proof}

Of course, \cref{thm:D-lower} is meaningful for cardinal mechanisms. Deterministic ordinal mechanisms can be easily seen to be at least $2$-approximate.


\section{Conclusion}\label{sec:open}
We have presented a series of positive and negative results for a simple hybrid social choice model, which combines elements of mechanism design with and without monetary transfers. Closing the gap between the approximation ratio of $5/4$ of the template mechanism $R$ (see \cref{sec:beyond}) and our general unconditional lower bound of approximately $1.14$ for any truthful mechanism (see \cref{sec:lower-bounds}) is an important and definitely non-trivial challenge. Besides this concrete open problem, there are many natural extensions of the model that are worth studying. For example, we have weighed equally the contribution of the expert and the agents to the social welfare. Generalizing the definition of the welfare by introducing a factor of $\alpha>0$, by which the contribution of the expert will be multiplied, is a first such extension.

Another extension could be to consider a different optimization objective, possibly by mixing the welfare of the expert with the revenue that can be extracted by the agents. The underlying mechanism design problem now seems to be quite different from the one we have studied here. For the revenue to be (part of) a meaningful objective, one would have to restrict attention to \emph{individually rational} mechanisms, which guarantee that the agents obtain non-negative utility. This is important, as otherwise a truthful mechanism could ignore their bids and charge them the maximum amount. In fact, the literature of revenue-maximization (e.g., see \citep{myerson1981optimal}) focuses on mechanisms which are individually rational. 

It is not hard to see however that in our problem, bid-independent, individually rational mechanisms always extract zero revenue. It is also well-documented that revenue maximization is a less meaningful objective in the absence of prior knowledge of the agents' values \citep{hartline2013mechanism} and it is commonly assumed that these values are drawn from some known distributions \citep{AGT,myerson1981optimal}. Hence, designing efficient truthful mechanisms for such an optimization objective requires radically different ideas, or perhaps even the migration to a Bayesian setting, like the one mentioned above.

Our model of one expert and two competing agents can be thought of as the simplest possible non-trivial hybrid social choice scenario. There are many important generalizations that one could consider for future research. Indicatively, these could include larger populations of experts and agents, more than one assets to be transferred with combinatorial constraints governing their acquisition, or even dynamic expert preferences that depend on the bidding information. 
These questions pose quite a few challenges. For example, even having two experts enlarges the space of possible truthful mechanisms significantly, as now other mechanisms, known as \emph{duples} and \emph{hierarchical unilaterals} come into play. Adding more agents seems more manageable, but our characterization of \cref{lem:ECh-IC} no longer applies, at least in its current form. The case of more than one assets seems even more challenging, because the setting on the agents' part is no longer single-parameter, and therefore we can not use Myerson's characterization from \cref{lem:BCh-IC}.

\section*{Acknowledgments}
We are grateful to Arunava Sen for fruitful discussions at initial stages of this paper. This work was partially supported by COST Action CA 16228 ``European Network for Game Theory'', by a PhD scholarship from the Onassis Foundation, by the ERC Advanced Grant 321171 (ALGAME), by the ERC Grant 639945 (ACCORD), by the Swiss National Science Foundation under contract number 200021\_165522, and by the IIT Kanpur Grant IITK/CS/2017198.


\bibliographystyle{named}
\bibliography{expert-bib}

\end{document}